\documentclass[12pt]{article}

\usepackage{hyperref}

\usepackage{amsfonts, amsmath, amssymb, amsthm}
\usepackage{a4}

\usepackage{relsize}

\DeclareMathOperator{\const}{const}

\DeclareMathOperator{\Sp}{Sp}
\DeclareMathOperator{\linearspan}{span}
\DeclareMathOperator{\sgn}{sgn}

\newtheorem{theorem}{Theorem}

\newtheorem{lemma}{Lemma}

\theoremstyle{definition}

\theoremstyle{remark}
\newtheorem{remark}{Remark}

\author{I.G. Korepanov}
\title{Bosonic pentachoron weights and \goodbreak multiplicative 2-cocycles}
\date{March--October 2017}

\begin{document}

\sloppy

\maketitle

\begin{abstract}
Gaussian pentachoron weights can be used for constructing algebraic realizations of four-dimen\-sional Pachner moves. Here, we consider a natural `gauge equivalence' for such weights with one and two bosonic---i.e., commuting---variables on 3-faces. For the one-boson case, all generic weights turn out to be gauge equivalent. For the two-boson case, and generic weights, their gauge equivalence classes are parameterized by multiplicative 2-cocycles. Moreover, a generic two-boson weight can be reduced by a gauge transformation to a delta-function form.
\end{abstract}

\section{Introduction}\label{s:i}

The aim of this paper is to shed some light on the algebraic nature of `Gaussian pentachoron weights' that can be used in constructing four-dimen\-sional topological quantum field theories. As explained below, the word `bosonic' in the title simply means that we are using \emph{commuting} variables, in contrast with the `fermionic' case and anti-commuting variables in papers~\cite{full-nonlinear,fermionic-1,fermionic-2}.

\subsection{Pentachoron weights and algebraic realizations of four-dimen\-sional Pachner moves}\label{ss:intro-1}

Pentachoron weight is an algebraic expression put in correspondence to a pentachoron (\,=\,4-simplex) in a triangulation of a piecewise linear (PL) four-manifold. These weights must be able to compose in some way when pentachora are glued together, so as a quantity belonging to the whole manifold could be obtained from them. This quantity can be called \emph{invariant} if it does not depend on a specific triangulation. As two triangulations can be transformed one into the other by a sequence of \emph{Pachner moves}~\cite{Pachner,Lickorish}---elementary triangulation re-buildings---there must exist some algebraic equalities involving pentachoron weights (or, more generally, $n$-simplex weights in the case of $n$-manifolds), corresponding naturally to Pachner moves. Such equalities are called \emph{algebraic realizations} of Pachner moves, or simply \emph{Pachner move relations}.

\subsection{Gaussian exponentials as pentachoron weights}\label{ss:intro-2}

A simple and natural idea is to try pentachoron weights in the form of \emph{Gaussian exponentials}---exponentials of quadratic forms---of some variables living on the faces of a pentachoron. Let these faces be \emph{3-faces}, i.e., tetrahedra. Composition of pentachoron weights corresponding to gluing pentachora together along some 3-faces will then be described by an integral w.r.t.\ the variables living on these 3-faces.

This idea works well in the \emph{fermionic} case, that is, with the mentioned variables being \emph{anticommutative} and belonging to a \emph{Grassmann algebra} (and the integral being, accordingly, the \emph{Berezin integral}). Moreover, fascinating mathematical structures have been unveiled in the fermionic case, namely, \emph{nonlinear} parameterization of the `essential part' of the mentioned quadratic forms in terms of the manifold's middle (simplicial) cohomologies~\cite{full-nonlinear,fermionic-1,fermionic-2}.

\subsection{Annihilating differential operators}\label{ss:intro-3}

We studied fermionic Gaussian exponentials (or, as we called them, \emph{Grassmann--Gaussian exponentials}) by means of their \emph{annihilating differential operators}. A Gaussian exponential~$\mathcal W$ depending on $n$ variables $x_1,\ldots, x_n$ is determined, to within a constant factor, by a system of differential equations of the following form:
\begin{equation}\label{eo:d}
\begin{cases}
\left(\dfrac{\partial}{\partial x_1} + l_1(x_1,\ldots, x_n)\right) \mathcal W=0,\\
\dotfill \\
\left(\dfrac{\partial}{\partial x_n} + l_n(x_1,\ldots, x_n)\right) \mathcal W=0,
\end{cases}
\end{equation}
where $l_1,\ldots,l_n$ are \emph{linear} forms of our~$x$'s.

We assume that each~$x_i$, \ $i=1,\ldots,n$, belongs to a certain 3-face~$t$ of the given pentachoron~$u$. In~\cite{full-nonlinear}, there was exactly one Grassmanian variable~$x_t$ on each~$t$, so $n=5$ for any pentachoron. In the present paper, we consider `usual' commutative variables, and two different cases: either one variable on each~$t$, or two variables on each~$t$ ($n=10$ in the latter case).

\subsection{Edge operators}\label{ss:intro-4}

The main tool in our work was \emph{edge operators}---a particular case of differential operators annihilating a pentachoron weight. Let $b\subset u$ be an edge belonging to pentachoron~$u$. Suppose we can make up such a linear combination of the operators in big parentheses in~\eqref{eo:d} that it contains (both in its differential part and in the linear form) only the $x$'s lying on the \emph{three tetrahedra containing~$b$}. Then this linear combination is called an edge operator corresponding to~$b$.

Fermionic edge operators appeared in paper~\cite{full-nonlinear}, and their key role was determined by the two following facts. First, there are of course linear dependencies between edge operators, and studying these dependencies reveals their underlying simplicial \emph{2-cocycle}~\cite[Subsection~3.3]{full-nonlinear}. Such cocycles parameterize, to within \emph{gauge transformations} (of which we speak below), all generic Grassmann--Gaussian exponents of the kind investigated in~\cite{full-nonlinear}, that is, with \emph{one} Grassmann variable~$x_t$ living on each triangulation tetrahedron (\,=\,3-face)~$t$.

Second, edge operators add together nicely when pentachora are glued together, and this was the basis for the proof of the Pachner move 3--3 relation in~\cite[Section~5]{full-nonlinear}.

\subsection{Gauge equivalence of pentachoron weights}\label{ss:intro-5}

Differential operators $\partial / \partial x_i$ and operators of multiplication by the variables~$x_i$ generate an algebra, called \emph{Clifford algebra} in the fermionic case and \emph{Heisenberg algebra} in the bosonic case. Such algebras have rich automorphism groups. In particular, there are automorphisms involving only the variable(s) on one single tetrahedron. Such automorphisms generate the subgroup of \emph{gauge transformations} in the automorphism group (for details, see~\cite[Definition~9]{full-nonlinear} in the fermionic case, or Subsection~\ref{ss:1g} and the beginning of Section~\ref{s:2} in the present paper in the bosonic case).

It turns out that gauge transformations preserve algebraic Pachner move relations. To be more exact, if two pentachora $u_1$ and~$u_2$ have a common 3-face~$t$ (that is, they have been glued along this 3-face), and if a transformation~$\tau$ involving only the variable(s) on~$t$ has been done for~$u_1$, then the transformation $\sigma\tau\sigma^{-1}$ must be done for~$u_2$. Here $\sigma$ changes the signs of all the differentiations, but leaves intact all multiplication operators:
\[
\sigma(\partial / \partial x_i) = -\partial / \partial x_i, \qquad \sigma(x_i) = x_i.
\]
If, on the other hand, a 3-face~$t$ is a \emph{boundary} face of both the initial and the final configurations of a Pachner move, then simply the same arbitrary gauge transformation, involving only the variable(s) on~$t$, can always be done for both configurations.

\subsection{Bosonic Gaussian exponentials and integrals}\label{ss:intro-6}

If we believe that (roughly speaking) whatever can be done for fermions, can be done for bosons as well (and vice versa), it makes sense to realize the ideas of~\cite{full-nonlinear} also in a bosonic setting. Here we need \emph{commuting} variables, for instance, simply \emph{real} variables, and we can use usual integrals over all real values of each variable for gluing pentachora together. That is, if we are gluing two pentachora (or two clusters of pentachora) along a 3-face~$t$ where the variable~$x_t$ lives, and their weights are $\mathcal W_1$ and~$\mathcal W_2$, then the resulting weight will be
\[
\int_{-\infty}^{\infty} \mathcal W_1 \mathcal W_2 \,\mathrm dx_t.
\]

It turns out more convenient to use \emph{formal} integrals, avoiding possible problems with their divergence. The value ascribed to such integral is obtained by the analytic continuation from those values of the relevant parameter where the integral converges. Actually, we must come to terms with the fact that there are \emph{two} values of the integral, differing in sign, due to the two-valuedness of the square root, as in the following example:
\begin{equation}\label{eo:int}
\int \exp(-ax^2)\,\mathrm dx\, \stackrel{\mathrm{def}}{=} \, \sqrt{\pi/a}\qquad \text{for any complex \ }a\ne 0.
\end{equation}
We need complex numbers in this work, because complex transformations of quadratic forms will be essential in Section~\ref{s:2}.

As we will see in Subsections \ref{ss:eo:cw} and~\ref{ss:eo:ar}, such integrals work well for Pachner moves 3--3, even though \emph{delta functions} may appear during the integration, as in the following example:
\begin{equation}\label{eo:ea}
\int \exp(axy)\,\mathrm dy = 2\pi \delta(\sqrt{-1}\,ax).
\end{equation}

We do not consider algebraic realizations of other (than 3--3) Pachner moves in this paper. It must be admitted, however, that these other moves cannot be treated in the same straightforward way, because there appear integrals that do not exist even in the mentioned formal sense (like $\int \mathrm dx$, or integral involving a delta function squared). Such integrals will require some `regularization'; a similar problem was solved in the fermionic case by introducing a \emph{chain complex}~\cite{fermionic-1} instead of just one matrix of a quadratic form. One more possible way is described below.

\subsection{Discrete bosonic variables}\label{ss:intro-7}

All problems with integrals happily disappear if our bosonic variables are `discrete'---take only a finite number of values. For instance, they may belong to a \emph{finite field}~$\mathbb F$. The exponential can then be replaced by any homomorphism
\[
\mathrm e\colon\;\,\mathbb F \to \mathbb C^*
\]
from the \emph{additive} group of~$\mathbb F$ into the \emph{multiplicative} group of~$\mathbb C$, and the integral---by the sum over all elements in~$\mathbb F$ for the relevant variable. This idea has been already realized in paper~\cite{cubic}, where the quadratic forms in Gaussian exponents were found using reductions $\mod 2$ of some formulas written initially for real variables---actually, some variation of formulas~\eqref{6b} below. 

Note that a similar construction---using the rings~$\mathbb Z_n$ instead of our~$\mathbb F$---has been succesfully realized in paper~\cite{kashaev}, where its author managed to build some realizations of all four-dimen\-sional Pachner moves. We would like to call those realizations `constant'---not depending on parameters (except for~$n$), in contrast with the case of~\cite{full-nonlinear}, where the realization depends on an arbitrary 2-cocycle. We will present here strong evidence that this constantness was not accidental: parameters, also in the form of a 2-cocycle---but this time \emph{multiplicative}---appear only when the number of bosons living at each 3-face is \emph{doubled}.

\subsection{What we do in this paper}\label{ss:intro-8}

As we have stated, we like to describe pentachoron weights in terms of their annihilating operators, and these are divided into gauge equivalence classes. Hence, it is natural to study pentachoron weights \emph{up to the gauge equivalence}, and this is what we are doing below.

We restrict ourself to the \emph{complex} field~$\mathbb C$ (not just~$\mathbb R$, because we want square roots to be always available; this will be especially needed in our Lemma~\ref{l:d}). Interestingly, our present work actually serves for the `discrete' case as well: as we have said already, our formulas---or their small variations---can admit meaningful reductions modulo a finite characteristic~\cite{cubic}.

Also, we always consider---and classify---\emph{generic} pentachoron weights. This looks, of course, especially natural for the field~$\mathbb C$, but we hope that our work can shed some light upon the finite characteristic case as well.

We mainly analyze here (gauge equivalence classes of) Gaussian exponentials for just one pentachoron. The number of `bosonic' complex variables at each 3-face is either one or two. For the one-boson case, all generic weights turn out to be gauge equivalent, and this suggests an explanation of the constantness of Pachner move realizations in~\cite{kashaev}. For the two-boson case, and generic weights, their gauge equivalence classes are parameterized by \emph{multiplicative} 2-cocycles. Moreover, a generic two-boson weight can be reduced by a gauge transformation to a simple delta function form (which proved extremely important, because `hexagon cohomologies' can then be introduced, see~\cite{cubic}).

We also pay some attention to the Pachner move 3--3, in order to explain how exactly it is ensured by the form of our edge operators, and show thus the relevance of our approach for piecewise linear topology.

Below,
\begin{itemize}\itemsep 0pt
 \item in Section~\ref{s:L}, we discuss some generalities about Gaussian exponentials,
 \item in Section~\ref{s:1}, we consider the case of a single boson living on each 3-face of a pentachoron, and find that almost all Gaussian pentachoron weights are gauge equivalent in this case,
 \item in Section~\ref{s:1P}, we show how the 3--3 Pachner move relation follows from a reasoning with edge operators, for the `single boson' case of Section~\ref{s:1},
 \item in Section~\ref{s:2}, we consider the two-boson case,
 \item and in Section~\ref{s:d}, we discuss our results and plans for further work.
\end{itemize}

\section{Gaussian exponentials and Lagrangian subspaces of operators}\label{s:L}

\emph{Gaussian exponential} is, by definition, the exponential of a quadratic form. More specifically, let $F$ be a symmetric $n\times n$ matrix, and $\mathsf x=\begin{pmatrix}x_1&\ldots&x_n\end{pmatrix}^{\mathrm T}$---an $n$-column of variables. We assume that both matrix entries and variables~$x_t$, \ $t=1,\ldots,n$ (we use the subscript~$t$ because our variables will belong to \emph{tetrahedra}), take values in the field~$\mathbb C$. Then, the Gaussian exponential
\begin{equation}\label{W}
\mathcal W = \exp (-\frac{1}{2}\,\mathsf x^{\mathrm T} F\, \mathsf x)
\end{equation}
obeys the following $n$ differential equations:
\begin{equation}\label{ann}
\left( \partial_t+\sum_{t'=1}^n F_{tt'}x_{t'} \right) \mathcal W = 0, \qquad t=1,\ldots,n,
\end{equation}
where we denote $\partial_t=\partial / \partial x_t$.

We now consider the $2n$-dimen\-sional symplectic space~$V$---the $\mathbb C$-linear span of all operators $\partial_t$ and~$x_t$ (where $x_t$ is interpreted, of course, as a multiplication operator), with the symplectic scalar product given by the commutator. In other words, $V$ consists of all linear operators of the form
\begin{equation}\label{bg}
d = \sum_{t=1}^n (\beta_t\partial_t+\gamma_t x_t).
\end{equation}
The linear span of operators in the parentheses in~\eqref{ann} is a \emph{Lagrangian subspace} in~$V$, that is, a maximal subspace with all commutators of its elements vanishing. We will also need the notions of a \emph{$t$-component} of an operator~\eqref{bg}:
\begin{equation}\label{bgt}
d|_t \stackrel{\mathrm{def}}{=} \beta_t\partial_t+\gamma_t x_t,
\end{equation}
and a \emph{partial scalar product} (\,=\,partial commutator) for a given~$t$:
\[
[d,d']_t \stackrel{\mathrm{def}}{=} [d|_t, d'|_t] = \beta_t \gamma'_t - \gamma_t \beta'_t.
\]
For a given~$t$, all operators~\eqref{bgt} form, of course, a two-dimen\-sional space; we call it~$V_t$, so that $V=\bigoplus_{t=1}^n V_t$.

Conversely, almost all Lagrangian subspaces in~$V$ determine a Gaussian exponential. We will need, however, also Lagrangian subspaces that determine not exactly a Gaussian exponential but a limit case of these containing Dirac delta functions. The simplest example arises when the Lagrangian subspace is spanned by all operators~$x_t$; in this case~$\mathcal W$ is (proportional to) $\prod_{t=1}^n \delta(x_t)$. A much more interesting example will appear in Subsection~\ref{ss:2d}.

\section{One boson on a 3-face: gauge equivalence of almost all pentachoron weights}\label{s:1}

In this Section, we work within just one pentachoron; we call it $u=12345$, where $1,\ldots,5$ are (the numbers of) its vertices. We attach to each of its 3-faces---tetrahedra~$t$---one complex variable~$x_t$, and a two-dimen\-sional symplectic linear space $V_t=\linearspan\{\partial_t,x_t\}$.

Pentachoron weight will have the form~\eqref{W}, with a $5\times 5$ symmetric matrix~$F$. We will use simplified notations for its entries~$F_{tt'}$, replacing $t$ and~$t'$ with vertices $u\setminus t$ and~$u\setminus t'$, respectively. For instance, we will write~$F_{12}$ instead of~$F_{2345,1345}$.

\subsection{Gauge equivalence classes, and formulation of this section's main result}\label{ss:1g}

We consider the natural actions of $\Sp(2,\mathbb C)$ on each of five symplectic spaces~$V_t$, \ $t\subset u$. Hence, the product $\Sp(2,\mathbb C)^{\times 5}$ of five groups acts on the direct sum $V_u = \bigoplus_{t\subset u}V_t$, and, obviously, it maps a Lagrangian subspace $V_L \subset V_u$ again into a Lagrangian subspace. All Lagrangian subspaces break up into equivalence classes---orbits of this action; we call this \emph{gauge} equivalence.

The main result of this section is the following theorem.

\begin{theorem}\label{th:1g}
Almost all matrices~$F$ determine Lagrangian subspaces belonging to one single orbit.
\end{theorem}

Below in Subsection~\ref{ss:1h} we explain heuristically why Theorem~\ref{th:1g} holds, and Subsections \ref{ss:1e} and~\ref{ss:1gi} contain its constructive proof.

\subsection{Heuristic parameter count}\label{ss:1h}

Symmetric $5\times 5$ matrix~$F$ contains 15 independent entries. Symplectic group $\Sp(2n,\mathbb C)$ contains 3 independent parameters for $n=1$. Hence, all $5\times 3=15$ parameters in~$F$ are gauge, and two Lagrangian subspaces produced by two generic symmetric matrices~$F$ as explained in Section~\ref{s:L} are gauge equivalent.

\subsection{Edge operators}\label{ss:1e}

For a given Lagrangian subspace $V_L \subset V_u$ and an edge~$b$, an \emph{edge operator}~$d_b$ belongs, by definition, to~$V_L$, and has the form
\begin{equation}\label{1:db}
d_b = \sum_{t\supset b}(\beta_{bt}\partial_t+\gamma_{bt} x_t)
\end{equation}
---that is, only the three tetrahedra containing~$b$ are involved. For a generic~$V_L$, this determines~$d_b$ up to a normalizing constant.

Choosing this constant in the most natural way, we get the following formulas for coefficients in~\eqref{1:db}:
\begin{align}
\beta_{bt} & =F_{ik}F_{jl}-F_{il}F_{jk}, \label{1:bt} \\
\gamma_{bt} & =F_{ik}F_{jl}F_{mm}-F_{il}F_{jk}F_{mm}-F_{ik}F_{jm}F_{lm} \nonumber \\
 & \qquad +F_{im}F_{jk}F_{lm}+F_{il}F_{jm}F_{km}-F_{im}F_{jl}F_{km}. \label{1:gt} 
\end{align}
Here, vertices $i,\ldots,m$ make a permutation of $1,\ldots,5$, satisfying the following conditions:
\begin{enumerate}\itemsep 0pt
 \item edge $b=ij$,
 \item tetrahedron $t=ijkl$,
 \item and the vertices in~$t$ must be written in such order that gives its orientation as induced from~$u=12345$. That is, if $b=12$, then the three tetrahedra containing~$b$ and contained in~$u$ must be written as $1234$, $1253$ and~$1245$.
\end{enumerate}

A generic~$V_L$ is exactly the linear span of the ten edge operators. Also, if all coefficients are as in \eqref{1:bt} and~\eqref{1:gt}, then the edge operators~\eqref{1:db} have the following directly checked properties:
\begin{itemize}\itemsep 0pt
 \item independence of edge orientation:
\begin{equation}\label{1:symm}
d_{ij}=d_{ji},
\end{equation}
 \item simple linear dependencies at each vertex~$i$:
\begin{equation}\label{1:1c}
\sum_{j\ne i} d_{ij}=0\quad \text{in each vertex }i,
\end{equation}
 \item equalness of $t$-components for non-intersecting edges of~$t$:
\[
d_{ij}|_{ijkl}=d_{kl}|_{ijkl}, \text{ \ that is, \ } \beta_{ij,ijkl}=\beta_{kl,ijkl},\quad \gamma_{ij,ijkl}=\gamma_{kl,ijkl},
\]
 \item and partial symplectic scalar products $[d_b,d_{b'}]_t$ all either vanishing or equal to each other, up to a possible sign. Namely, a nonzero product appears iff $b$ and~$b'$ both belong to~$t$ and, moreover, share a vertex, call it~$i$:
\begin{equation}\label{1:c}
[d_{ij},d_{ik}]_{ijkl}=\pm c.
\end{equation}
In~\eqref{1:c}, the plus sign appears if the orientation of $t=ijkl$ induced from the pentachoron corresponds to the order~$ijkl$ of its vertices, otherwise the minus sign is used. As for the quantity~$c$, it has the following remarkable form: up to a coefficient~$1/2$, $c$ consists of the same terms that enter in the determinant of~$F$, but only those that correspond to \emph{cyclic} permutations:
\begin{multline}\label{iae}
c = \frac{1}{2} \sum_{\mathrm{cyclic}\;\pi} \sgn \pi \, F_{1\pi(1)}F_{2\pi(2)}F_{3\pi(3)}F_{4\pi(4)}F_{5\pi(5)}\\
=F_{12}F_{13}F_{24}F_{35}F_{45}-F_{12}F_{14}F_{23}F_{35}F_{45}-F_{12}F_{13}F_{25}F_{34}F_{45}+F_{12}F_{15}F_{23}F_{34}F_{45}\\
+F_{13}F_{14}F_{23}F_{25}F_{45}-F_{13}F_{15}F_{23}F_{24}F_{45}+F_{12}F_{14}F_{25}F_{34}F_{35}-F_{12}F_{15}F_{24}F_{34}F_{35}\\
-F_{13}F_{14}F_{24}F_{25}F_{35}+F_{14}F_{15}F_{23}F_{24}F_{35}+F_{13}F_{15}F_{24}F_{25}F_{34}-F_{14}F_{15}F_{23}F_{25}F_{34}.
\end{multline}
\end{itemize}

\begin{remark}
Note the difference with the fermionic case~\cite{full-nonlinear}, where \eqref{1:1c} also holds if edge operators are normalized properly, but property~\eqref{1:symm} is replaced with the \emph{anti}symmetry. The reader can see in~\cite[Subsection~3.3]{full-nonlinear} how it leads to the appearance of an \emph{additive} 2-cocycle---but in our present case, it does \emph{not} appear!
\end{remark}

\begin{remark}
But a cocycle must be somewhere, if we believe in the analogy between bosons and fermions. And indeed, it will appear in our Section~\ref{s:2}, but strikingly enough, it will be \emph{multiplicative}!
\end{remark}

\subsection[Gauge isomorphism for two different~$V_L$, and a canonical form for matrix~$F$]{Gauge isomorphism for two different~$\boldsymbol{V_L}$, and a canonical form for matrix~$\boldsymbol F$}\label{ss:1gi}

It follows from~\eqref{1:c} that operators $(d_{12}|_{1234})/\sqrt{c}$ and~$(d_{13}|_{1234})/\sqrt{c}$ form a \emph{canonical basis} in~$V_{1234}$. In a two-dimen\-sional space, this simply means that the symplectic scalar product of two basis vectors is one. Similarly, $(-d_{12}|_{1235})/\sqrt{c}$ and~$(d_{13}|_{1235})/\sqrt{c}$ form a canonical basis in~$V_{1235}$, and so on. This gives us the obvious way of constructing a gauge automorphism of~$V_u=V_{1234}\oplus\ldots\oplus V_{2345}$ that sends the Lagrangian subspace corresponding to one matrix~$F$ into the Lagrangian subspace corresponding to another matrix. Namely, for each of the five~$t$, the element in $\Sp(2,\mathbb C)$ acting in~$V_t$ must map one such basis into the other, and the gauge automorphism is the direct sum of five such elements.

We have thus proved Theorem~\ref{th:1g}.

\smallskip

In particular, we can consider such matrix~$F$ that the mentioned canonical bases take the following simplest form: $\partial_{1234}$ and $x_{1234}$ in~$V_{1234}$, then $\partial_{1235}$ and $-x_{1235}$ in~$V_{1235}$, and, in general, as follows: \emph{the basis in~$V_{ijkl}$ consists of $\partial_{ijkl}$ and $\pm x_{ijkl}$, where $i<j<k<l$, and the plus sign is taken if the tetrahedron orientation determined by the order~$ijkl$ of its vertices coincides with the orientation induced from the pentachoron; otherwise the minus is taken.} Such~$F$ can be regarded as `canonical', and a direct calculation shows that it looks as follows:
\begin{equation}\label{cF}
F = \begin{pmatrix}
  1 & 0 & -1 &  1 & 0 \\
  0 & 0 & 0 & -1 &  1 \\
  -1 & 0 & 1 & 0 & -1 \\
   1 & -1 & 0 & 0 & 0 \\
  0 &  1 & -1 & 0 & 1
\end{pmatrix}.
\end{equation}

We will see in Subsection~\ref{ss:1E} how to construct an algebraic realization of Pachner move 3--3 using matrices~\eqref{cF}.

It makes sense to write out the explicit form of the edge operators for matrix~\eqref{cF}. It follows from the above and from the fact that $c=1$ for it, where $c$ is the expression~\eqref{iae}, that, for a tetrahedron~$t=ijkl$, \ $i<j<k<l$,
\begin{equation}\label{et}
d_{ij}|_t=d_{kl}|_t=\partial_t,\qquad d_{ik}|_t=d_{jl}|_t=\pm x_t,\qquad d_{il}|_t=d_{jk}|_t=-\partial_t\mp x_t,
\end{equation}
where the upper sign is taken if the orientation~$ijkl$ coincides with the one induced from the pentachoron; otherwise the lower sign is taken.

\section{One boson on a 3-face: the 3--3 relation}\label{s:1P}

\subsection{The 3--3 Pachner move}\label{ss:33}

The 3--3 Pachner move transforms a cluster of three pentachora grouped around their common 2-face~$s$, and forming what is called in PL topology the \emph{star} of~$s$, into another cluster, occupying the same place in a manifold triangulation, and also consisting of three pentachora. We use below the following terminology: the initial cluster is called the \emph{left-hand side} (l.h.s.) of the move, while the final cluater---its \emph{right-hand side} (r.h.s.). We orient all pentachora consistently; when this orientation coincides with the one determined by the (increasing) order of pentachoron vertices, we denote the pentachoron simply by its vertices, and to mark the opposite orientation, we use the wide tilde.

The l.h.s. of our 3--3 move will consist of pentachora $12345$, $\widetilde{12346}$ and~$12356$. There are the following \emph{inner tetrahedra}: $1234$, $1235$, and~$1235$.

As for the r.h.s., it consists of pentachora $12456$, $\widetilde{13456}$ and~$23456$. The inner tetrahedra are $1456$, $2456$ and~$3456$.

Both sides of the move have, of course, the same boundary, consisting of nine tetrahedra.

Below, we denote the weight of a pentachoron $ijklm$ as~$\mathcal W_{ijklm}$, and the weight of pentachoron $\widetilde{ijklm}$ as~$\widetilde{\mathcal W}_{ijklm}$.

One more comment that will be important for Lemma~\ref{l:1:r} below is that there are the same fifteen edges in both sides of the move 3--3, and all these edges are \emph{boundary}.

\subsection{Cluster weights and formal integrals}\label{ss:eo:cw}

First, we introduce the following notion of \emph{quasi-Gaussian weight}. In brief, this is a multidimensional delta function multiplied by a Gaussian exponential. To be exact, let there be some variables~$x_t$ (in our context, they correspond to tetrahedra~$t$); we combine them, for a moment, into a vector~$\vec x$. Let $\ell_1(\vec x),\ldots,\ell_m(\vec x)$ be \emph{linearly independent linear forms} (case $m=0$ is acceptable), and $\mathcal A(\vec x)$ a quadratic form. Quasi-Gaussian weight is, by definition, an expression of the form
\begin{equation}\label{eo:qG}
\const \cdot \prod_{i=1}^m \delta\bigl(\ell_i(\vec x)\bigr) \cdot \exp \mathcal A(\vec x) .
\end{equation}

Suppose that individual pentachoron weights have the form~\eqref{eo:qG} (which is certainly so for Gaussian exponentials~\eqref{W}). Making a cluster of pentachora leads to taking the product of weights, and then some integrals (as in formula~\eqref{1:WWW} below. Recall that, as we explained in Subsection~\ref{ss:intro-6}, our integrals are \emph{formal}, and their definition is clear from the examples \eqref{eo:int} and~\eqref{eo:ea}). Assume that, in doing so, we do \emph{not} run into
\begin{enumerate}\itemsep 0pt
 \item\label{i:eo:p} a product of delta functions of \emph{linearly dependent} arguments,
 \item\label{i:eo:c} an integral of the form $\int \const \mathrm dx$, where `$\const$' means an expression not depending on~$x$ (and not an identical zero).
\end{enumerate}
Then, it is easily seen that the cluster weight is well defined (up to a sign, due to possible square roots arising from integrations), and has again the quasi-Gaussian form~\eqref{eo:qG}.

\begin{remark}
Besides the move 3--3, there are, in four dimensions, also Pachner moves 2--4 and 1--5, and in the r.h.s.'s of these, there arise products of delta functions mentioned in item~\ref{i:eo:p} which cannot be further integrated. We have already discussed this problem and its possible solutions in Subsections \ref{ss:intro-6} and~\ref{ss:intro-7}.
\end{remark}

Cluster weights obey differential equations, as we will see in Subsection~\ref{ss:1E}, much like weight~\eqref{W} obeys~\eqref{ann}. So, \emph{partial derivatives} of quasi-Gaussian weights come into play. These may contain, compared with~\eqref{eo:qG}, also factors linear in variables~$x_t$, and derivatives of delta functions. This of course by no means impedes integrating them in the same formal way as we did it above.

There are two easily checked facts about partial derivatives and formal integrals that will be of use for us in Lemma~\ref{l:eo:r} below. We formulate thses facts as the following Lemma.

\begin{lemma}\label{l:eo:k}
 \begin{enumerate}
  \item\label{i:eo:k1} Let $\mathcal W$ be a quasi-Gaussian weight. Then integration w.r.t.\ one variable commutes with the differentiation w.r.t.\ another variable:
\begin{equation}\label{eo:id}
 \frac{\partial}{\partial x_t} \int \mathcal W \,\mathrm dx_{t'} =  \int \frac{\partial \mathcal W}{\partial x_t} \,\mathrm dx_{t'},
\end{equation}
$t\ne t'$, if both sides of~\eqref{eo:id} exist.
  \item\label{i:eo:k2} Integral of the partial derivative of a quasi-Gaussian weight w.r.t.\ the integration variable vanishes:
\begin{equation*}
\int \frac{\partial \mathcal W}{\partial x_t} \,\mathrm dx_t = 0.
\end{equation*}
 \end{enumerate}
\qed
\end{lemma}

\subsection{The algebraic realization of move 3--3}\label{ss:eo:ar}

\begin{theorem}\label{th:1:33}
The following algebraic realization of Pachner move 3--3 holds for $\mathcal W_u$ of the form~\eqref{W} with $F$ given by~\eqref{cF}, and $\widetilde{\mathcal W}_u$ also of the form~\eqref{W}, but with $F$ given by \emph{minus} the r.h.s.\ of~\eqref{cF}.
\begin{multline}\label{1:WWW}
\iiint \mathcal W_{12345} \widetilde{\mathcal W}_{12346} \mathcal W_{12356} \,\mathrm dx_{1234} \,\mathrm dx_{1235} \,\mathrm dx_{1236} \\
 = \const \iiint \mathcal W_{12456} \widetilde{\mathcal W}_{13456} \mathcal W_{23456} \,\mathrm dx_{1456} \,\mathrm dx_{2456} \,\mathrm dx_{3456}.
\end{multline}
\end{theorem}

We mean of course that, when constructing the weight for a pentachoron $ijklm$, \ $i<i<k<l<m$, we use the formulas from Section~\ref{s:1} but make vertex substitutions $1\mapsto i$, \ldots, $5\mapsto m$.

\begin{remark}
The integration variables in~\eqref{1:WWW} are exactly those corresponding to the inner tetrahedra in the respective side of the move, listed in Subsection~\ref{ss:33}. The result in either side of~\eqref{1:WWW} depends on the nine \emph{boundary} variables.
\end{remark}

\begin{proof}[Proof of Theorem~\ref{th:1:33}]
A direct calculcation shows that both sides of~\eqref{1:WWW} equal
\begin{multline*}
\const \cdot\, \delta(x_{1245}-x_{1246}-x_{1345}+x_{1346}) \, \delta(x_{1356}-x_{2356}-x_{1346}+x_{2346}) \\
\cdot \exp\bigl( -(x_{1356}-x_{1346}-x_{1256}+x_{1246})(x_{2346}-x_{2345}-x_{1346}+x_{1345}) \bigr),
\end{multline*}
and even the constants are here the same.
\end{proof}

This proof does not explain, however, \emph{why} \eqref{1:WWW} holds, which is of course important if we want to find further Pachner move realizations. We do this below in Subsection~\ref{ss:1E}. The fundamental role will be played, of course, by edge operators.

\subsection{Cluster edge operators}\label{ss:1E}

We now introduce the notion of edge operators not only for an individual pentachoron, but for a \emph{cluster} of pentachora such as the l.h.s.\ or r.h.s.\ of move 3--3. Although we will not use any other clusters here, we remark that, actually, our definition is meaningful if we understand under the name of `cluster' any \emph{triangulated orientable PL 4-manifold with boundary}.

First, we re-write formula~\eqref{1:db} in such way that it will refer explicitly to the relevant pentachoron~$u$:
\begin{equation}\label{1:dbu}
d_b^{(u)} = \sum_{\substack{t\supset b\\ t\subset u}} (\beta_{bt}^{(u)}\partial_t + \gamma_{bt}^{(u)}x_t),
\end{equation}
while the cluster edge operator that we are constructing will be denoted~$\mathsf d_b$.

One key requirement enabling our construction is as follows. Let a tetrahedron~$t$ belong to both pentachoron~$u_1$ and pentachoron~$u_2$, then we require that
\begin{equation}\label{1:bg}
\beta_{bt}^{(u_1)}=\beta_{bt}^{(u_2)},\qquad \gamma_{bt}^{(u_1)}=-\gamma_{bt}^{(u_2)}.
\end{equation}
\emph{Note that this holds indeed for matrix~$F$ given by~\eqref{cF}}, because $F$ changes its sign for the pentachora with wide tildes. This all agrees also with the general definition of `canonical' basis in the paragraph before~\eqref{cF}.

Given an edge~$b$, consider all \emph{boundary} tetrahedra $t\supset b$. Each such~$t$ belongs to a single pentachoron~$u$, and we simply take the $t$-component of the pentachoron edge operator~\eqref{1:dbu} for the $t$-component of our cluster edge operator, while for any inner (non-boundary)~$t$ we set its $t$-component to be zero:
\begin{equation}\label{1:bi}
\mathsf d_b|_t \stackrel{\mathrm{def}}{=} \begin{cases}d_b^{(u)}|_t & \text{for a boundary }t, \\ 0 & \text{for an inner }t. \end{cases}
\end{equation}
Below, we ignore the inner tetrahedra, and consider operators~$\mathsf d_b$ as sums of their components~\eqref{1:bi} over only boundary tetrahedra. It is then clear that each~$\mathsf d_b$ is \emph{the same for the l.h.s.\ and r.h.s.\ of the Pachner move}---enough to refer to the explicit form~\eqref{et}.

\begin{lemma}\label{l:eo:r}
For any edge~$b$ in the l.h.s.\ or r.h.s.\ of move 3--3, both sides of~\eqref{1:WWW} are annihilated by the edge operator\/~$\mathsf d_b$.
\end{lemma}

\begin{proof}
Choose one side of the Pachner move, for instance, its l.h.s. Consider an edge operator~$\mathsf d_{ij}$, and let there be, for definiteness, exactly one inner tetrahedron~$t=ijkl$ (in the mentioned l.h.s.) containing edge~$ij$ (other cases are analyzed in a similar way). Denote the two pentachora containing~$t$ as $u_1$ and~$u_2$, and their respective weights as $\mathcal W_1$ and~$\mathcal W_2$. Now do the following calculation, taking into account \eqref{1:dbu} and~\eqref{1:bg}:
\begin{multline*}
\beta_{ij,t}^{(u_1)}\partial_t(\mathcal W_1 \mathcal W_2) = \mathcal W_2 \beta_{ij,t}^{(u_1)}\,\partial_t\mathcal W_1 + \mathcal W_1\beta_{ij,t}^{(u_1)}\,\partial_t \mathcal W_2 \\
 = \mathcal W_2 \left( -\gamma_{ij,t}^{(u_1)}x_t - \sum \bigl(\text{components of }d_b^{(u_1)}\text{ for tetrahedra}\ne t \bigr)\right)\mathcal W_1 \\
 + \mathcal W_1 \left( \gamma_{ij,t}^{(u_1)}x_t - \sum \bigl(\text{components of }d_b^{(u_2)}\text{ for tetrahedra}\ne t \bigr)\right)\mathcal W_2 \\
 = -\sum \bigl(\text{all components of }d_b^{(u_2)} \text{ and }d_b^{(u_2)}\text{ for tetrahedra}\ne t \bigr)(\mathcal W_1 \mathcal W_2) \\
 = -\mathsf d_b (\mathcal W_1 \mathcal W_2).
\end{multline*}
So, $\mathsf d_b (\mathcal W_1 \mathcal W_2)$ is a partial derivative w.r.t.~$x_t$. Remember that $\mathcal W_1$ and~$\mathcal W_2$ are two of the three weights in the l.h.s.\ of~\eqref{1:WWW}, while the third weight~$\mathcal W_3$ does depends neither on any variables entering in~$\mathsf d_b$ nor on~$x_t$. Thus, also
\begin{equation}\label{eo:pd}
\mathsf d_b (\mathcal W_1 \mathcal W_2 \mathcal W_3) = \partial_t \bigl(-\beta_{ij,t}^{(u_1)}\mathcal W_1 \mathcal W_2 \mathcal W_3 \bigr).
\end{equation}
Finally, $x_t$ is one of the variables $x_{1234}$, $x_{1235}$ or $x_{1236}$, with respect to which the triple integral in the l.h.s.\ of~\eqref{1:WWW} is taken. It remains to integrate both sides of~\eqref{eo:pd} w.r.t.~$x_t$, and apply item~\ref{i:eo:k1} of Lemma~\ref{l:eo:k} to the l.h.s., and item~\ref{i:eo:k2} of Lemma~\ref{l:eo:k} to the r.h.s.
\end{proof}

Recall now that the l.h.s.\ and r.h.s.\ of the move 3--3 contain the same fifteen edges, and all of these belong to the boundary.

\begin{lemma}\label{l:1:r}
A function~$\mathcal F$ of nine variables~$x_t$, where $t$ runs over all boundary tetrahedra of either the l.h.s.\ or the r.h.s.\ of move 3--3, is determined uniquely, up to a multiplicative constant, by the requirement that it has a quasi-Gaussian form~\eqref{eo:qG} and is annihilated by all the fifteen edge operators\/~$\mathsf d_b$.\end{lemma}

\begin{proof}
A computer calculation shows that the mentioned edge operators span a nine-dimen\-sional linear space which, in its turn, contains a two-dimen\-sional subspace~$V_{\mathrm{pure}\;x}$ of ``pure $x$'s''---operators without a differential part. Now we do a linear transformation of our variables---express the $x$'s linearly in terms of new variables~$y_i$, \ $i=1,\ldots,9$, in such way that $V_{\mathrm{pure}\;x}$ is spanned simply by $y_1$ and~$y_2$. This means that~$\mathcal F$ contains a two-dimen\-sional delta function factor that can be written as $\delta(y_1)\delta(y_2)$. There remain seven more usual equations for the rest of~$y_i$, determining uniquely the Gaussian exponential factor.
\end{proof}

Lemmas \ref{l:eo:r} and~\ref{l:1:r} provide thus a `conceptual' way of proving Theorem~\ref{th:1:33}, where we don't need to calculate explicitly the l.h.s.\ and r.h.s.\ of~\eqref{1:WWW}.

\section{Two bosons on a 3-face}\label{s:2}

We consider again, as in Section~\ref{s:1}, one pentachoron~$u=12345$. But now we attach to each of its 3-faces~$t$ \emph{two} complex variables $x_t$ and~$y_t$, and the following symplectic space~$V_t$:
\[
V_t=\linearspan \left\{ \frac{\partial}{\partial x_t}, \frac{\partial}{\partial y_t}, x_t, y_t \right\},
\]
with the symplectic scalar product given by the commutator.

Thus, symplectic group $\Sp(4,\mathbb C)$ acts in each~$V_t$, and the product $\Sp(4,\mathbb C)^{\times 5}$ acts in $V_u=\bigoplus_{t\subset u} V_t$. We call the elements of this product \emph{gauge transformations}, and orbits of this action \emph{gauge equivalence classes}, in full analogy with Subsection~\ref{ss:1g}.

Generic pentachoron weight has again the form~\eqref{W}, where $F$ is now a $10\times 10$ symmetric matrix, and $\mathsf x=\begin{pmatrix}x_1&y_1&\dots&x_5&y_5\end{pmatrix}$. Operators in the big parentheses in~\eqref{ann} span a 10-dimensional Lagrangian subspace in the 20-dimensional space~$V_u$.

\subsection{Main results of this section}\label{ss:2m}

The main results of this Section are as follows:
\begin{enumerate}\itemsep 0pt
 \item\label{i:c} to within gauge equivalence, generic pentachoron weights of the form~\eqref{W} are parameterized by multiplicative 2-cocycles,
 \item\label{i:d} generic pentachoron weight is reduced, by a gauge transformation, to a product of five delta functions of linear combinations of $x_t$ and~$y_t$,
 \item\label{i:e} for a cocycle that \emph{never takes value~$-1$}, a corresponding pentachoron weight \emph{always} exists; in particular, it can be chosen in the delta function form of item~\ref{i:d},
 \item\label{i:33}relation 3--3 for generic weights.
\end{enumerate}

\begin{remark}
Note that the above items \ref{i:c} and~\ref{i:d} speak about the \emph{general position}, meaning `some unspecified Zariski open set' of weights. In contrast with these, item~\ref{i:e} gives the exact sufficient condition that guarantees that a corresponding weight can be constructed.
\end{remark}

\begin{remark}
The mentioned exact sufficient condition may be important for passing on to the \emph{finite characteristic} case, see discussion in Section~\ref{s:d}.
\end{remark}

\subsection{Parameter count}\label{ss:2h}

Symmetric $10\times 10$ matrix~$F$ contains 55 independent entries. Symplectic group $\Sp(4,\mathbb C)$ contains 10 independent parameters. The product of five such groups contains 50 parameters, but one of these does not affect the exponential---we will see it in the end of Subsection~\ref{ss:2g}. Hence, $55-49=6$ parameters in~$F$ are not gauge, and they correspond to a multiplicative 2-cocycle.

\subsection{Edge operators and their linear dependencies}\label{ss:2e}

For a given Lagrangian subspace $V_L \subset V_u$ and an edge~$b$, an edge operator $d_b$
belongs, by definition, to $V_L$, and is a linear combination of only those $\partial/\partial x_t$, $\partial/\partial y_t$, $x_t$ and~$y_t$, for which $t\supset b$.

For a generic~$F$, a two-dimensional linear space~$E_b$ of edge operators corresponds to each edge~$b$. We choose some bases $\mathsf e_b, \mathsf f_b$ in these spaces, with the understanding that
\begin{equation}\label{ijji}
\mathsf e_{ij} = \mathsf e_{ji}, \qquad \mathsf f_{ij} = \mathsf f_{ji}.
\end{equation}
There are two linear dependencies at each vertex, as, for instance, in vertex~1:
\begin{equation}\label{1}
A_{12} \begin{pmatrix}\mathsf e_{12} \\ \mathsf f_{12} \end{pmatrix} +
A_{13} \begin{pmatrix}\mathsf e_{13} \\ \mathsf f_{13} \end{pmatrix} +
A_{14} \begin{pmatrix}\mathsf e_{14} \\ \mathsf f_{14} \end{pmatrix} +
A_{15} \begin{pmatrix}\mathsf e_{15} \\ \mathsf f_{15} \end{pmatrix} = 0 ,
\end{equation}
with some $2\times 2$ matrices~$A_{12}$, \ldots, $A_{15}$. These dependencies arise because the eight operators $\mathsf e_{12},\ldots,\mathsf f_{15}$ span only a six-dimen\-sional subspace of~$V_L$, namely, subspace of operators not containing $x_{2345}$, $y_{2345}$, $\partial/\partial x_{2345}$, and $\partial/\partial y_{2345}$. For a generic~$F$, matrices~$A_{12}$, \ldots, $A_{15}$ are nondegenerate, and determined up to the left multiplication by an arbitrary nondegenerate matrix, one and the same for all of them. Due to dependence~\eqref{1}, we can define natural isomorphisms between $E_{12}$, $E_{13}$, $E_{14}$ and~$E_{15}$, acting on the basis vectors as follows:
\begin{equation}\label{1213}
A_{12} \begin{pmatrix}\mathsf e_{12} \\ \mathsf f_{12} \end{pmatrix} \leftrightarrow
A_{13} \begin{pmatrix}\mathsf e_{13} \\ \mathsf f_{13} \end{pmatrix} \leftrightarrow
A_{14} \begin{pmatrix}\mathsf e_{14} \\ \mathsf f_{14} \end{pmatrix} \leftrightarrow
A_{15} \begin{pmatrix}\mathsf e_{15} \\ \mathsf f_{15} \end{pmatrix}.
\end{equation}

We call isomorphisms~\eqref{1213} ``isomorphisms at vertex~$1$''. One of them maps basis $\begin{pmatrix}\mathsf e_{13} \\ \mathsf f_{13} \end{pmatrix}$ of~$E_{13}$ into basis $A_{13}^{-1} A_{12} \begin{pmatrix}\mathsf e_{12} \\ \mathsf f_{12} \end{pmatrix}$ of~$E_{12}$. We can define then, in a similar way, isomorphisms at vertex~$3$. One of these maps basis $\begin{pmatrix}\mathsf e_{23} \\ \mathsf f_{23} \end{pmatrix}$ of~$E_{23}$ into basis $A_{32}^{-1} A_{31} \begin{pmatrix}\mathsf e_{13} \\ \mathsf f_{13} \end{pmatrix}$ of~$E_{13}$. We can also define isomorphisms at vertex~$2$, and one of them maps basis $\begin{pmatrix}\mathsf e_{12} \\ \mathsf f_{12} \end{pmatrix}$ of~$E_{12}$ into basis $A_{21}^{-1} A_{23} \begin{pmatrix}\mathsf e_{23} \\ \mathsf f_{13} \end{pmatrix}$ of~$E_{23}$.

Combining these three isomorphisms, we get an \emph{automorphism} of~$E_{12}$ corresponding to going around triangle~$123$, given by
\begin{equation}\label{123}
\begin{pmatrix}\mathsf e_{12} \\ \mathsf f_{12} \end{pmatrix} \mapsto A_{21}^{-1} A_{23} A_{32}^{-1} A_{31} A_{13}^{-1} A_{12} \begin{pmatrix}\mathsf e_{12} \\ \mathsf f_{12} \end{pmatrix} .
\end{equation}

We can now define also similar automorphisms of~$E_{12}$ corresponding to going around triangles $124$ and~$125$, or even more complex trajectories made of pentachoron edges. Amazingly, all such automorphisms \emph{commute}, so they can, in the general case, be diagonalized simultaneously! Moreover, the two eigenvalues for any such automorphism are mutually inverse. We will see all this in Theorem~\ref{th:c} below. We would like, however, to state right here what follows immediately from these facts.

First, we can choose, for the automorphism corresponding to each triangle~$s$, \emph{one} of these eigenvalues, in a consistent way for all~$s$, so that if we denote these eigenvalues as~$\omega_s$, then they form a \emph{multiplicative 2-cocycle}:
\begin{equation}\label{cc}
\omega_s=\frac{1}{\omega_{-s}},\qquad\qquad \frac{\omega_{ijk}\omega_{ikl}}{\omega_{ijl}\omega_{jkl}}=1.
\end{equation}
Here `$-s$' means the same triangle~$s$, but with the opposite orientation.

\begin{theorem}\label{th:c}
\begin{enumerate}\itemsep 0pt
 \item\label{i:ca} For a generic~$F$, bases\/ $\mathsf e_b, \mathsf f_b$ in spaces~$E_b$ can be chosen so that \emph{all} matrices~$A_{ij}$, \ $1\le i,j\le 5$, \ $i\ne j$, can be taken diagonal. Matrices~$A_{ij}$ are of course defined according to~\eqref{1}, with an obvious change of subscripts.
 \item\label{i:cb} Moreover, if $\mathsf e_b,\mathsf f_b \in E_b$ and $\mathsf e_{b'},\mathsf f_{b'} \in E_{b'}$ are such bases for any two edges $b$ and~$b'$, then $[\mathsf e_b, \mathsf e_{b'}]_t=0$ and $[\mathsf f_b, \mathsf f_{b'}]_t=0$ for any individual tetrahedron~$t$.
 \item\label{i:cc} The two eigenvalues of automorphisms of type~\eqref{123} are mutually inverse.
\end{enumerate}
\end{theorem}

\begin{proof}

First, we prove a few lemmas.

\begin{lemma}\label{l:1234}
To determine matrices $A_{12}$, $A_{13}$ and~$A_{14}$ in~\eqref{1}, it is enough to take the components of edge operators for edges 12, 13 and~14, corresponding to the space $V_{1234} = \linearspan\{x_{1234},y_{1234},\partial/\partial x_{1234},\partial/\partial y_{1234}\}$. That is, they are determined from the following relation:
\begin{equation}\label{1234}
A_{12} \begin{pmatrix}\mathsf e_{12}|_{1234} \\ \mathsf f_{12}|_{1234} \end{pmatrix} +
A_{13} \begin{pmatrix}\mathsf e_{13}|_{1234} \\ \mathsf f_{13}|_{1234} \end{pmatrix} +
A_{14} \begin{pmatrix}\mathsf e_{14}|_{1234} \\ \mathsf f_{14}|_{1234} \end{pmatrix} = 0.
\end{equation}
\end{lemma}

\begin{remark}
Remember that these matrices are determined to within a left multiplication by a nondegenerate matrix.
\end{remark}

\begin{remark}
The same matrices $A_{12}$ and~$A_{13}$---taken up to a left multiplication, so it may be better to speak about the product $A_{12}^{-1}A_{13}$---appear, of course, also from the tetrahedron~$1235$.
\end{remark}

\begin{proof}[Proof of Lemma~\ref{l:1234}]
This follows from the fact that $V_{1234}$ is four-dimen\-sional, so there must be two dependencies between $\mathsf e_{12}|_{1234},\ldots,\mathsf f_{14}|_{1234}$, and they must obviously give the same matrices as in~\eqref{1}, except for~$A_{15}$ that is irrelevant for tetrahedron~$1234$.
\end{proof}

\begin{lemma}\label{l:k}
Let $\mathsf d_b$, for each edge~$b$, be an arbitrary linear combination of $\mathsf e_b$ and~$\mathsf f_b$---that is, an arbitrary element of~$E_b$. Then,
\begin{enumerate}\itemsep 0pt
 \item\label{i:ka} if $b$ and~$b'$ are nonintersecting edges in tetrahedron~$t$, then $[\mathsf d_b,\mathsf d_{b'}]_t=0$,
 \item\label{i:kb} if $b$ and~$b'$ have a common vertex, and $t$ and~$t'$ are two tetrahedra, each containing both $b$ and~$b'$, then $[\mathsf d_b,\mathsf d_{b'}]_t+[\mathsf d_b,\mathsf d_{b'}]_{t'}=0$. For instance, $[\mathsf d_{12},\mathsf d_{13}]_{1234}+[\mathsf d_{12},\mathsf d_{13}]_{1235}=0$.
\end{enumerate}
\end{lemma}

\begin{proof}
The full commutator of edge operators always vanishes: $[\mathsf d_b,\mathsf d_{b'}]=0$. In the case~\ref{i:ka}, it is composed of a single summand, corresponding to tetrahedron~$t$. In the case~\ref{i:kb}, it is composed of two summands, corresponding to tetrahedra $t$ and~$t'$.
\end{proof}

We now consider the following two-dimen\-sional linear subspaces in~$V_{1234}$, corresponding to edges of tetrahedron~$1234$: for each edge~$ij\subset 1234$, we take the space spanned by the $1234$-components of $\mathsf e_{ij}$ and~$\mathsf f_{ij}$. According to item~\ref{i:kb} of Lemma~\ref{l:k}, subspaces corresponding to edges $34$, $24$ and~$23$ are the \emph{orthogonal complements} of subspaces corresponding to edges $12$, $13$ and~$14$, respectively. Taking a suitable basis in~$V_{1234}$, we can identify vectors in~$V_{1234}$ with 4-rows~$\mathsf v$ in such way that the 1234-component of the symplectic scalar product $\langle \mathsf u, \mathsf v \rangle$ will be written as
\begin{equation}\label{sc}
\mathsf u \begin{pmatrix} 0&1&0&0\\ -1&0&0&0\\ 0&0&0&1\\ 0&0&-1&0\end{pmatrix} \mathsf v^{\mathrm T}.
\end{equation}

A generic two-dimen\-sional linear subspace in~$V_{1234}$ is then the span of the rows of a $2\times 4$ matrix of the form
\begin{equation}\label{B}
\mathsf B = \begin{pmatrix}\begin{matrix}1&0\\0&1\end{matrix} & \mathlarger{\mathlarger{\mathlarger {B}}} \end{pmatrix}.
\end{equation}

\begin{lemma}\label{l:d}
Matrices~$B$ in~\eqref{B} can be made diagonal for all six edges 12, \ldots, 34 simultaneously, by a change of basis in~$V_{1234}$ preserving the form~\eqref{sc} of the scalar product.
\end{lemma}

\begin{proof}
If $M$ and~$N$ are two $2\times 2$ matrices with $\det M=\det N=1$, then  matrix $\mathsf B$~\eqref{B} can always be transformed as follows:
\begin{equation}\label{t}
\mathsf B \mapsto M \mathsf B \begin{pmatrix} M^{-1}&0\\ 0&N \end{pmatrix},
\end{equation}
without changing the form~\eqref{sc} of the symplectic product. Moreover, $B$ can be chosen as zero matrix for one edge, say~12. Then, an easy exercise, using the fact that a generic square matrix over the field~$\mathbb C$ can always be diagonalized, shows that matrices~$B$ for edges 13 and~14 can be made diagonal by a transformation~\eqref{t} with properly chosen $M$ and~$N$.

So, matrices~$\mathsf B$ for edges 12, 13 and~14 acquire the form
\begin{equation}\label{f}
\mathsf B = \begin{pmatrix} *&0&*&0\\ 0&*&0&* \end{pmatrix}.
\end{equation}
Then, matrices~$B$ for the rest of edges have the same form~\eqref{f} automatically, because, as we already mentioned above, each of their corresponding spaces is the orthogonal complement to the space spanned by the rows of a matrix of the form~\eqref{f}, with respect to the symplectic scalar product~\eqref{sc}.
\end{proof}

It follows from Lemma~\ref{l:d} that we can \emph{re-define} our basis vectors $\mathsf e_b$ and~$\mathsf f_b$ for all edges $b\subset 1234$ in such way that their $1234$-components acquire the forms
\begin{equation}\label{e}
\mathsf e_b|_{1234}=\begin{pmatrix}*&0&*&0\end{pmatrix},\qquad \mathsf f_b|_{1234}=\begin{pmatrix}0&*&0&*\end{pmatrix}
\end{equation}
in a proper basis in~$V_{1234}$. We assume now that we \emph{have re-defined} them this way. It can also be checked that, generically, there are no proportional vectors among vectors~\eqref{e} for $b=12$, $13$ and~$14$. It follows then from Lemma~\ref{l:1234} that we can assume that all three matrices $A_{12}$, $A_{13}$ and~$A_{14}$ are diagonal: indeed, we can always put, say, $A_{12}$ equal to identity matrix, then $A_{13}$ and~$A_{14}$ become necessarily diagonal because $\mathsf e_{12}|_{1234}$ can be a linear combination only of $\mathsf e_{13}|_{1234}$ and $\mathsf e_{14}|_{1234}$, and similarly for $\mathsf f_{12}|_{1234}$.

The same reasoning with vertex~$1$ replaced first by~$2$ and then by~$3$ shows, in particular, that matrices $A_{21}$, $A_{23}$, $A_{31}$ and~$A_{32}$ can all be also taken diagonal, for the \emph{same} basis vectors $\mathsf e_b, \mathsf f_b$. So, for our basis vectors chosen as in~\eqref{e}, all matrices in~\eqref{123} become diagonal.

Now, it becomes an easy exercise to show that actually \emph{all} matrices~$A_{ij}$, \ $1\le i,j\le 5$, \ $i\ne j$, can be made diagonal simultaneously by choosing proper bases $\mathsf e_b, \mathsf f_b$ in all spaces~$E_b$, and this proves item~\ref{i:ca} in Theorem~\ref{th:c}. And from now on, we \emph{assume that we have re-defined all} $\mathsf e_b$ and~$\mathsf f_b$ this way.

Item~\ref{i:cb} in Theorem~\ref{th:c} is now proved as follows. For $t=1234$, it holds due to the form~\eqref{e} of edge operator components. We explain how to prove it, for instance, for $t=1235$ and operators~$\mathsf e_b$.

First, the $1235$-component of the scalar product between any two of $\mathsf e_{12}$, $\mathsf e_{13}$ and~$\mathsf e_{23}$ vanishes due to item~\ref{i:kb} in Lemma~\ref{l:k} (and vanishing of the $1234$-component). Next, in our situation, any of $\mathsf e_{35}|_{1235}$, $\mathsf e_{25}|_{1235}$ and~$\mathsf e_{15}|_{1235}$ is a linear combination of $\mathsf e_{12}|_{1235}$, $\mathsf e_{13}|_{1235}$ and~$\mathsf e_{23}|_{1235}$---this follows from the analogues of~\eqref{1234} for tetrahedron~$1235$ and vertices 1, 2 and~3, and the diagonality of matrices~$A_{ij}$, and this proves~\ref{i:cb}.

Finally, item~\ref{i:cc} in Theorem~\ref{th:c} holds due to the following lemma.

\begin{lemma}\label{l:rm}
The determinant of automorphism~\eqref{123} equals~$1$.
\end{lemma}

\begin{proof}
Direct computer calculation.
\end{proof}

\begin{remark}
In this paper, Lemma~\ref{l:rm} is the only one that has no conceptual proof as yet.
\end{remark}

Lemma~\ref{l:rm} completes the proof of Theorem~\ref{th:c}.

\end{proof}

\subsection[Gauge equivalence of pentachoron weights with the same~$\omega$]{Gauge equivalence of pentachoron weights with the same~$\boldsymbol{\omega}$}\label{ss:2g}

Recall that multiplicative 2-cocycle~$\omega$ was introduced in the paragraph before Theorem~\ref{th:c}.

\begin{theorem}\label{th:g}
Pentachoron weights with the same~$\omega$ are gauge equivalent.
\end{theorem}

\begin{proof}
We must fix some `standard' bases in all~$V_t$, so that all edge operators have components w.r.t.\ these bases depending only on~$\omega$.

Theorem~\ref{th:c} shows how to choose edge operators $\mathsf e_b$ and~$\mathsf f_b$, and recall that we have chosen them (between Lemmas \ref{l:d} and~\ref{l:rm}) exactly that way. There still remains some arbitrariness: $\mathsf e_b$ and~$\mathsf f_b$ may be multiplied by any constants. We will fix most of this arbitrariness by imposing a special form on matrices~$A_{ij}$. Recall that they have the form
\begin{equation}\label{Aij}
A_{ij}=\begin{pmatrix}\gamma_{ij}&0\\ 0&\gamma_{ij}^{-1}\end{pmatrix}.
\end{equation}
We set
\begin{equation}\label{gm}
\gamma_{ij}=\begin{cases}\nu_{ij}&\text{if \ }i<j,\\ 1&\text{if \ }i>j.\end{cases}
\end{equation}
Here $\nu$ is such a multiplicative 1-chain that $\delta \nu=\omega$. This means that $\nu_{ij}=\nu_{ij}^{-1}$, and $\omega_{ijk}=\nu_{ij}\nu_{ik}^{-1}\nu_{jk}$. Of course, $\nu_{ij}$ never vanishes. Also, we will assume that $\nu$ has been built from~$\omega$ in some `standard' way, for instance,
\[
\nu_{ij}=1 \text{ \ for \ } i=1, \text{ \ and \ } \nu_{ij}=\omega_{1ij} \text{ \ for \ } i,j\ne 1.
\]

A small exercise shows that matrices~$A_{ij}$ can always be brought into the form \eqref{Aij}, \eqref{gm} using the mentioned arbitrariness in edge operators.

We have linear dependencies
\begin{equation}\label{h}
\gamma_{12}\mathsf e_{12}+\dots+\gamma_{15}\mathsf e_{15}=0
\end{equation}
and so on. Note that all such dependencies \emph{fix the normalization of all operators~$\mathsf e_b$ to within one overall common factor}.

Due to the dependencies of type~\eqref{h}, a linear combination
\begin{equation}\label{lc}
\lambda_{12}\mathsf e_{12}|_{1234}+\dots+\lambda_{34}\mathsf e_{34}|_{1234}
\end{equation}
will vanish if the row
\[
\begin{pmatrix} \lambda_{12} & \lambda_{13} & \lambda_{14} & \lambda_{23} & \lambda_{24} & \lambda_{34} \end{pmatrix}
\]
is proportional to any row of matrix
\begin{equation}\label{g4}
\begin{pmatrix} \gamma_{12} & \gamma_{13} & \gamma_{14} & 0 & 0 & 0 \\
                \gamma_{21} & 0 & 0 & \gamma_{23} & \gamma_{24} & 0 \\
                0 & \gamma_{31} & 0 & \gamma_{32} & 0 & \gamma_{34} \\
                0 & 0 & \gamma_{41} & 0 & \gamma_{42} & \gamma_{43} \end{pmatrix}.
\end{equation}

Recall~\eqref{e} that all operators~$\mathsf e_b|_{1234}$ span a two-dimen\-sional space. If $\mathrm e_{\xi}^{(1234)}$ and $\mathrm e_{\eta}^{(1234)}$ denote two basis vectors in it, then we can write
\begin{equation}\label{b}
\mathsf e_b|_{1234} = \xi_b^{(1234)} \mathrm e_{\xi}^{(1234)} + \eta_b^{(1234)} \mathrm e_{\eta}^{(1234)},
\end{equation}
and call $\xi_b^{(1234)}$ and~$\eta_b^{(1234)}$ the \emph{components} of~$\mathsf e_b|_{1234}$. Writing~\eqref{lc} in these components, we find that a column made of either $\xi_{12}^{(1234)},\dots,\xi_{34}^{(1234)}$ or $\eta_{12}^{(1234)},\dots,\eta_{34}^{(1234)}$ must be orthogonal to the rows of matrix~\eqref{g4}. These columns must not be proportional, but otherwise they can be chosen arbitrarily---this corresponds to the arbitrariness of basis $\mathrm e_{\xi}^{(1234)},\mathrm e_{\eta}^{(1234)}$. For instance, they can be written as follows:
\begin{equation}\label{6b}
\begin{pmatrix}\xi_{12}^{(1234)}\\ \xi_{13}^{(1234)}\\ \xi_{14}^{(1234)}\\ \xi_{23}^{(1234)}\\ \xi_{24}^{(1234)}\\ \xi_{34}^{(1234)} \end{pmatrix}=
\begin{pmatrix}0\\[.4ex]
\nu_{14} ( \nu_{23} \nu_{34}+\nu_{24}) \\[.4ex]
-\nu_{13} ( \nu_{23} \nu_{34}+\nu_{24}) \\[.4ex]
-\nu_{24} ( \nu_{13} \nu_{34}+\nu_{14}) \\[.4ex]
\nu_{23} ( \nu_{13} \nu_{34}+\nu_{14}) \\[.4ex]
\nu_{13} \nu_{24}-\nu_{14} \nu_{23}\end{pmatrix},
\qquad
\begin{pmatrix}\eta_{12}^{(1234)}\\ \eta_{13}^{(1234)}\\ \eta_{14}^{(1234)}\\ \eta_{23}^{(1234)}\\ \eta_{24}^{(1234)}\\ \eta_{34}^{(1234)} \end{pmatrix}=
\begin{pmatrix}-\nu_{13} ( \nu_{23} \nu_{34}+\nu_{24}) \\[.4ex]
\nu_{12} ( \nu_{23} \nu_{34}+\nu_{24}) \\[.4ex]
0\\[.4ex]
\nu_{13} \nu_{34}-\nu_{12} \nu_{24}\\[.4ex]
\nu_{12} \nu_{23}+\nu_{13}\\[.4ex]
-\nu_{12} \nu_{23}-\nu_{13}\end{pmatrix},
\end{equation}
\emph{if the following condition holds:}
\begin{equation}\label{ne-1}
\omega_s\ne-1 \text{ \ for all triangles }s.
\end{equation}
Condition~\eqref{ne-1} arises here because we meet in~\eqref{6b} expressions like $\nu_{23} \nu_{34}+\nu_{24}=(\omega_{234}+1)\nu_{24}$ which we want not to vanish. This condition looks important, and it will arise below again and again on various occasions.

\begin{remark}
Nevertheless, we consider in paper~\cite{cubic} the case $\omega_s \equiv 1$ in \emph{finite characteristics}. Interestingly, even in characteristic two where this means the same as $\omega_s \equiv -1$, the constructions of~\cite{cubic} work well.
\end{remark}

In the same way as we have done for $t=1234$, we can fix similar bases $\mathrm e_{\xi}^{(t)},\mathrm e_{\eta}^{(t)}$ for other tetrahedra~$t=ijkl$, \ $i<j<k<l$, and write out all $t$-components $\xi_{ij}^{(ijkl)}$ and~$\eta_{ij}^{(ijkl)}$ for all operators~$\mathsf e_b$. Then we must choose the remaining two basis vectors in each~$V_{ijkl}$. We do it as follows: we do the same construction as above in this Subsection, but for vectors~$\mathsf f_b$ instead of~$\mathsf e_b$, and replacing, accordingly, all~$\nu_{ij}$ with~$\nu_{ij}^{-1}$. We call these new vectors $\mathrm f_{\xi}^{(t)}$ and~$\mathrm f_{\eta}^{(t)}$.

`Standard' bases in all~$V_{ijkl}$ are almost ready. The remaining arbitrariness is as follows. As we have mentioned after formula~\eqref{h}, there is one overall factor for vectors~$\mathsf e_b$; for a given~$\mathcal W$, all~$\mathsf e_b$ can be multiplied by that factor without affecting the above construction. There is of course a similar factor for all~$\mathsf f_b$. As we will see now, these two factors are not independent, because, clearly, their \emph{product} enters in the \emph{commutators}~$[\mathsf e_b,\mathsf f_{b'}]$.

We are going to see that the commutators between basis vectors also have `standard' values, depending only on~$\nu_{ij}$ and one overall factor, but otherwise independent of a specific pentachoron weight~$\mathcal W$.

\begin{lemma}\label{l:comm}
If $\omega_s\ne -1$ for all triangles~$s$, then all commutators $[\mathrm e_{\ldots}^{(t)}, \mathrm f_{\ldots}^{(t)}]$, where the dots must be replaced arbitrarily by $\xi$ or/and~$\eta$, are determined by values~$\nu_{ij}$ uniquely up to an overall common factor.
\end{lemma}

\begin{proof}
Due to the linear dependencies
\[
\sum_{j=1}^5 \gamma_{ij}\mathsf e_{ij}=0,\qquad \sum_{j=1}^5 \gamma_{ij}^{-1}\mathsf f_{ij}=0,\qquad i=1,\ldots,5,
\]
for edge operators, there are also dependencies for commutators. We can begin with considering commutators
\begin{equation}\label{cm}
[\mathsf e_b, \mathsf f_{12}]_{1234}, \qquad b\subset 1234.
\end{equation}
As such commutator for $b=34$ vanishes, there are essentially five of them, and there are four linear dependencies at vertices $1,\ldots,4$. Hence, there remains one degree of freedom: the vector made of five commutators~\eqref{cm} is proportional to a vector explicitly expressed through~$\gamma_{ij}$. For $\gamma_{ij}$ as in~\eqref{gm}, it looks as follows:
\begin{multline*}
\begin{pmatrix} [\mathsf e_{12}, \mathsf f_{12}]_{1234} & [\mathsf e_{13}, \mathsf f_{12}]_{1234} & [\mathsf e_{14}, \mathsf f_{12}]_{1234} & [\mathsf e_{23}, \mathsf f_{12}]_{1234} & [\mathsf e_{24}, \mathsf f_{12}]_{1234} \end{pmatrix} = \const \cdot \\
\begin{pmatrix} \nu_{14}\nu_{23}-\nu_{13}\nu_{24} & \nu_{12}\nu_{24}+\nu_{14} & -\nu_{12}\nu_{23}-\nu_{13} & -\nu_{12}\nu_{24}-\nu_{14} & \nu_{12}\nu_{23}+\nu_{13} \end{pmatrix}.
\end{multline*}
As, for instance, $\nu_{12}\nu_{23}+\nu_{13}=(\omega_{123}+1)\nu_{13}$, the vector components are not all zero.

Then we can find similarly that there is only one degree of freedom for all commutators $[\mathsf e_b, \mathsf f_{b'}]_{1234}$, and then extend this to all tetrahedra using item~\ref{i:kb} in Lemma~\ref{l:k}.

As our basis vectors $\mathrm e_{\ldots}^{(t)}$ and~$\mathrm f_{\ldots}^{(t)}$ can be expressed as linear combinations of $\mathsf e_b|_t$ and~$\mathsf f_b|_t$ (using \eqref{b} and~\eqref{6b}), with coefficients expressed through~$\nu_{ij}$, the statement about one degree of freedom holds for them as well.
\end{proof}

We can now choose some `standard' normalization for commutators---for instance, set one of them to equal~1. All other commutators depend then only on~$\omega$. Recall that there are also two overall factors for $\mathsf e_b$ and~$\mathsf f_b$, but their product is fixed by our choice of standard commutator normalization. So, there remains one free parameter, and this yields a one-para\-metric family of gauge transformations that does not change~$\mathcal W$, already mentioned in Subsection~\ref{ss:2h}.

We have thus obtained `standard' bases in all~$V_t$. The components of our gauge isomorphism for two different~$F$ with the same~$\omega$ map, by definition, one such  `standard' basis into the other, for each~$V_t$.
\end{proof}

\subsection[Constructing a delta function pentachoron weight from a given~$\omega$t]{Constructing a delta function pentachoron weight from a given~$\boldsymbol{\omega}$}\label{ss:2d}

In the previous Subsection, operators $\mathrm e_{\ldots}^{(t)}$ and~$\mathrm f_{\ldots}^{(t)}$ were obtained from a given Gaussian pentachoron weight. Now we would like to \emph{set} the first half of these operators equal to multiplication operators:
\[
\mathrm e_{\xi}^{(t)}=x_t,\qquad \mathrm e_{\eta}^{(t)}=y_t,
\]
and the second half equal to pure differentiations, and \emph{construct} the corresponding weight~$\mathcal W$.

Our edge operators~$\mathsf e_b$ are
\begin{equation}\label{Axy}
\begin{pmatrix} \mathsf e_{12} \\ \vdots \\ \mathsf e_{45} \end{pmatrix} = A \begin{pmatrix} x_{2345} \\ y_{2345} \\ \vdots \\ x_{1234} \\ y_{1234} \end{pmatrix},
\end{equation}
where matrix~$A$ is made from columns~\eqref{6b} and similar columns for tetrahedra other than~$1234$ as follows:\\[\smallskipamount]
$A=$ \vadjust{\nobreak\vspace*{-\smallskipamount}}
\begin{align*}
\mbox{\footnotesize $\displaystyle
\begin{pmatrix}0 &\! 0 &\! 0 &\! 0 &\! \xi_{12}^{(1245)} &\! \eta_{12}^{(1245)} &\! -\xi_{12}^{(1235)} &\! -\eta_{12}^{(1235)} &\! \xi_{12}^{(1234)} &\! \eta_{12}^{(1234)}\\[.8ex]
0 &\! 0 &\! -\xi_{13}^{(1345)} &\! -\eta_{13}^{(1345)} &\! 0 &\! 0 &\! -\xi_{13}^{(1235)} &\! -\eta_{13}^{(1235)} &\! \xi_{13}^{(1234)} &\! \eta_{13}^{(1234)}\\[.8ex]
0 &\! 0 &\! -\xi_{14}^{(1345)} &\! -\eta_{14}^{(1345)} &\! \xi_{14}^{(1245)} &\! \eta_{14}^{(1245)} &\! 0 &\! 0 &\! \xi_{14}^{(1234)} &\! \eta_{14}^{(1234)}\\[.8ex]
0 &\! 0 &\! -\xi_{15}^{(1345)} &\! -\eta_{15}^{(1345)} &\! \xi_{15}^{(1245)} &\! \eta_{15}^{(1245)} &\! -\xi_{15}^{(1235)} &\! -\eta_{15}^{(1235)} &\! 0 &\! 0\\[.8ex]
\xi_{23}^{(2345)} &\! \eta_{23}^{(2345)} &\! 0 &\! 0 &\! 0 &\! 0 &\! -\xi_{23}^{(1235)} &\! -\eta_{23}^{(1235)} &\! \xi_{23}^{(1234)} &\! \eta_{23}^{(1234)}\\[.8ex]
\xi_{24}^{(2345)} &\! \eta_{24}^{(2345)} &\! 0 &\! 0 &\! \xi_{24}^{(1245)} &\! \eta_{24}^{(1245)} &\! 0 &\! 0 &\! \xi_{24}^{(1234)} &\! \eta_{24}^{(1234)}\\[.8ex]
\xi_{25}^{(2345)} &\! \eta_{25}^{(2345)} &\! 0 &\! 0 &\! \xi_{25}^{(1245)} &\! \eta_{25}^{(1245)} &\! -\xi_{25}^{(1235)} &\! -\eta_{25}^{(1235)} &\! 0 &\! 0\\[.8ex]
\xi_{34}^{(2345)} &\! \eta_{34}^{(2345)} &\! -\xi_{34}^{(1345)} &\! -\eta_{34}^{(1345)} &\! 0 &\! 0 &\! 0 &\! 0 &\! \xi_{34}^{(1234)} &\! \eta_{34}^{(1234)}\\[.8ex]
\xi_{35}^{(2345)} &\! \eta_{35}^{(2345)} &\! -\xi_{35}^{(1345)} &\! -\eta_{35}^{(1345)} &\! 0 &\! 0 &\! -\xi_{35}^{(1235)} &\! -\eta_{35}^{(1235)} &\! 0 &\! 0\\[.8ex]
\xi_{45}^{(2345)} &\! \eta_{45}^{(2345)} &\! -\xi_{45}^{(1345)} &\! -\eta_{45}^{(1345)} &\! \xi_{45}^{(1245)} &\! \eta_{45}^{(1245)} &\! 0 &\! 0 &\! 0 &\! 0\end{pmatrix}.
$}
\end{align*}
We added the minus signs to the columns corresponding to tetrahedra $1345$ and~$1235$---that is, those whose orientation determined by the increasing order of their vertices does not coincide with the orientation induced from the pentachoron~$12345$. We did it keeping in mind that we will want to glue pentachora together, and to ensure the right agreement between the $t$-components of edge operators belonging to the two pentachora surrounding tetrahedron~$t$, as in~\eqref{1:bg}: the differential parts are the same, even if they vanish in our case, while the ``$x$-parts'' differ in signs.

In order that the space spanned by~$\mathsf e_b$ be five-dimen\-sional, the rank of matrix~$A$ must be~5, so there must exist a nonzero $5\times 5$ minor. Calculating a typical---namely, bottom right---minor yields
\[
\nu_{13}^{2} ( \nu_{12} \nu_{23}+\nu_{13})^{2} ( \nu_{12} \nu_{24}+\nu_{14})  ( \nu_{23} \nu_{34}+\nu_{24})  ( \nu_{23} \nu_{35}+\nu_{25}).
\]
So, we again need condition~\eqref{ne-1} (as, recall, $\nu_{12} \nu_{23}+\nu_{13}=(\omega_{123}+1)\nu_{13}$ and so on).

The space spanned by differential operators~$\mathsf f_b$ can be \emph{defined} now as containing all linear combinations of $\partial/\partial_{x_t}$ and $\partial/\partial_{y_t}$ that \emph{commute} with all~$\mathsf e_b$. This can be also put as follows. Look again at the r.h.s.\ of~\eqref{Axy}, but interpret each $x_t$ and~$y_t$ not as an operator, but simply as a complex variable. The r.h.s.\ of~\eqref{Axy} is made, in this interpretation, of ten linear functions---call them~$\phi_b(x_{2345},\ldots,y_{1234})$---of these variables. Then any~$\mathsf f_b$ must \emph{annihilate} all~$\phi_b(x_{2345},\ldots,y_{1234})$.

The pentachoron weight is now easily seen to be proportional to the product of delta functions:
\begin{equation}\label{dW}
\mathcal W = \const \prod \delta\bigl(\phi_b(x_{2345},\ldots,y_{1234})\bigr),
\end{equation}
taken over any five linearly independent~$\phi_b$.

\subsection{The 3--3 relation}\label{ss:33d}

The algebraic realization of Pachner move 3--3 in terms of delta function weights~\eqref{dW} looks as follows:
\begin{multline}\label{2:WWW}
\iiint \mathcal W_{12345} \mathcal W_{12346} \mathcal W_{12356} \,\mathrm dx_{1234} \,\mathrm dy_{1234} \,\mathrm dx_{1235} \,\mathrm dy_{1235} \,\mathrm dx_{1236} \,\mathrm dy_{1236} \\
 = \const \iiint \mathcal W_{12456} \mathcal W_{13456} \mathcal W_{23456} \,\mathrm dx_{1456} \,\mathrm dy_{1456} \,\mathrm dx_{2456} \,\mathrm dy_{2456} \,\mathrm dx_{3456} \,\mathrm dy_{3456}.
\end{multline}
Here is what must be explained about~\eqref{2:WWW}.

First, a multiplicative cocycle~$\omega$ is given, that is, values~$\omega_s$ for all triangles~$s$ in either side of the Pachner move, satisfying~\eqref{cc}. We assume that~$\omega$ is \emph{generic}, in particular, condition~\eqref{ne-1} holds. Second, the whole construction that has led us to the weight $\mathcal W$~\eqref{dW}, now denoted as~$\mathcal W_{12345}$, is repeated for each weight~$\mathcal W_{ijklm}$ in~\eqref{2:WWW}, with the obvious substitution of indices $12345 \mapsto ijklm$. Third, for our delta function case, as against the relation~\eqref{1:WWW}, there is no need to introduce weights with the wide tildes for `differently oriented' pentachora.

\begin{theorem}\label{th:2:WWW}
The equality~\eqref{2:WWW} holds indeed.
\end{theorem}

\begin{proof}
This can be proved using edge operators, along the same lines as we did in Subsection~\ref{ss:1E}. It can be checked, using computer algebra, that both sides of~\eqref{2:WWW} are, generically, \emph{nine-dimen\-sional} delta functions. Then, we can use only cluster edge operators having \emph{no differential part}, composed from the pentachoron edge operators of the form~\eqref{Axy} in the same way as in~\eqref{1:bi}. Like in Subsection~\ref{ss:1E}, they are the same for the l.h.s.\ and r.h.s.\ of the move. What remains is one more computer calculation to show that they span also a nine-dimen\-sional linear space.
\end{proof}

\section{Discussion of results}\label{s:d}

\subsection{Fermionic and bosonic cases}\label{ss:fb}

We have shown in~\cite{full-nonlinear} that an adequate tool for constructing four-dimen\-sional Pachner move realizations is brought about by studying Gaussian pentachoron weights using isotropic spaces of differential operators annihilating these weights. In~\cite{full-nonlinear,fermionic-1,fermionic-2} we did it for fermionic (Grassmann--Gaussian) weights, while in the present paper, we study the bosonic case.

We see that there are 2-cocycles hidden within both fermionic and bosonic Gaussian pentachoron weights. There appears to be, however, a striking difference between these two cases: the 2-cocycles are \emph{additive} in the fermionic case, while they are \emph{multiplicative} in the bosonic case.

\subsection{Transition to finite characteristic}\label{ss:fc}

Of interest are variants of our bosonic construction for fields of \emph{finite} characteristic. Namely, we can use formal differential operators over such a field~$\mathbb F$, but the pentachoron weight itself may remain \emph{complex} if we use, instead of our usual exponential, a homomorphism from the \emph{additive} group of~$\mathbb F$ into the \emph{multiplicative} group~$\mathbb C^*$.

\subsection{Nontriviality}\label{ss:ntr}

What we can expect to obtain, in the simplest one-boson case, from the construction proposed in Subsection~\ref{ss:fc}, can be seen by comparing our Section~\ref{s:1} with the paper~\cite{kashaev}. Namely, in the light of our results, it is very likely that, essentially, all `discrete Gaussian weights' are equivalent to one another, be they written in the form~\cite[(1)--(4)]{kashaev}, or in our canonical form~\eqref{cF}, or somehow else. Even in this case, however, the calculations in~\cite{kashaev} show the nontriviality of the corresponding four-manifold invariant.

As the two-boson case involves also a 2-cocycle---essentially its cohomology class, as shown in~\cite[Section~7]{fermionic-2} for the fermionic case---and is thus fundamentally richer than the one-boson case, it definitely deserves further studying. Moreover, even in the case of the trivial---identical unity---cocycle, there exist very nontrivial `hexagon cohomologies' for discrete analogues of delta function weights! These cohomologies lead to a very intriguing topological quantum field theory on four-dimen\-sional PL manifolds~\cite{cubic}.

\subsection{Interesting algebraic expression}\label{ss:iae}

Finally, the reader may find interesting our quantity~\eqref{iae}, containing \emph{some}, but not all, terms of the determinant of a symmetric $5\times 5$ matrix~$F$---namely those that correspond to \emph{cyclic} permutations.

\end{document}